\documentclass[twocolumn]{autart}    
\makeatletter
\newif\if@restonecol
\makeatother

\usepackage[linesnumbered,ruled,lined,commentsnumbered,norelsize]{algorithm2e}
\SetAlFnt{\scriptsize}
\SetAlCapNameFnt{\scriptsize}
\usepackage{algpseudocode}
\usepackage{graphicx}          

\usepackage[]{todonotes}
\usepackage[T1]{fontenc}
\let\theoremstyle\relax
\usepackage{amsfonts}
\usepackage{amsthm, amscd}
\usepackage{amsmath}
\usepackage{harvard} 
\usepackage{amssymb}
\usepackage{picins}
\usepackage[nodisplayskipstretch]{setspace}
\usepackage{epstopdf}
\usepackage{url}
\usepackage{lmodern}
\usepackage{xcolor}
\usepackage{ifthen}
\usepackage{float}
\usepackage[tight,footnotesize]{subfigure}
\usepackage[caption=false,font=footnotesize]{subfig}
\usepackage{booktabs}  %
\usepackage{diagbox}
\usepackage{caption}
\let\emptyset\varnothing
\usepackage{enumitem}
\usepackage{mathtools}
\usepackage{nccmath}
\usepackage{etoolbox}

\newtheorem{theorem}{Theorem}
\newtheorem{assumption}{Assumption}
\newtheorem{lemma}{Lemma}
\newtheorem{corollary}{Corollary}
\newtheorem{remark}{Remark}
\newtheorem{definition}{Definition}

\def\dh{\hat{d}}
\def\cA{{\mathcal A}}
\def\cG{{\mathcal G}}
\def\cF{{\mathcal F}}
\def\cC{{\mathcal C}}

\def\cS{{\mathcal S}}
\def\cD{{\mathcal D}}
\def\cK{{\mathcal K}}
\def\cU{{\mathcal U}}

\def\cI{{\mathcal I}}

\def\cN{{\mathcal N}}
\def\cL{{\mathcal L}}

\newcommand{\pump}[2]{
	\ifthenelse{\equal{#2}{}}{g}{
		\ifthenelse{\equal{#1}{0}}{#2}{
			\ifthenelse{\equal{#1}{1}}{g(#2)}{
				g^{#1}(#2)
			}
		}
	}
}
\newcommand{\pumpdir}[2]{
	\ifthenelse{\equal{#1}{0}}{#2}{
		\ifthenelse{\equal{#1}{1}}{K #2 + 1}{
			\ifthenelse{\equal{#1}{2}}{K^2 #2 + K + 1}{
				\ifthenelse{\equal{#1}{3}}{K^3 #2 + K^2 + K + 1}{
					K^{#1} #2 + \sum_{i=0}^{i<#1} K^i
				}
			}
		}
	}
}
\newcommand{\pumpinv}[2]{
	\ifthenelse{\equal{#1}{0}}{#2}{
		\ifthenelse{\equal{#1}{1}}{\left\lceil \frac{#2 - 1}{K} \right\rceil}{
			\ifthenelse{\equal{#1}{2}}{\left\lceil \frac{#2 - K - 1}{K^2} \right\rceil}{
				\ifthenelse{\equal{#1}{3}}{\left\lceil \frac{#2 - K^2 - K- 1}{K^3} \right\rceil}{
					\left\lceil \frac{#2 - \sum_{i=0}^{i<#1} K^i}{K^{#1}} \right\rceil
				}
			}
		}
	}
}

\begin{document}
\def\dh{\hat{d}}
\def\ph{{\hat{p}}}
\newcommand{\bff}{{\boldsymbol{f}}}
\newcommand{\bg}{{\boldsymbol{g}}}
\newcommand{\bI}{{\boldsymbol{I}}}
\newcommand{\bu}{{\boldsymbol{u}}}
\newcommand{\bD}{{\boldsymbol{D}}}
\newcommand{\bv}{{\boldsymbol{v}}}
\newcommand{\bS}{{\boldsymbol{S}}}
\newcommand{\bCO}{{\boldsymbol{CO}}}
\newcommand{\bc}{{\boldsymbol{c}}}	
\newcommand{\bF}{{\boldsymbol{F}}}	
\newcommand{\bm}{{\boldsymbol{m}}}	
\def\QED{~\rule[-1pt]{5pt}{5pt}\par\medskip}
\def\ll{{\lambda}}
\def\aa{{\alpha}}
\def\ee{{\epsilon}}
\def\DD{{\Delta}}
\def\TT{{\Theta}}
\def\tht{{\theta}}
\def\non{\nonumber}
\def\te{{\theta}}
\def\tte{{\tilde{\theta}}}
\def\tom{{\tilde{\omega}}}
\def\be{{\bar{e}}}
\def\lb{\label}
\def\ap{\thickapprox}
\def\bi{\begin{itemize}}
\def\ba{\begin{assumption}}
\def\ea{\end{assumption}}
\def\ee{\end{equation}}
\def\ei{\end{itemize}}
\def\de{\delta}
\def\om{\omega}
\def\eps{\epsilon}
\def\xh{\hat{x}}
\def\xu{{\bf x}}
\def\xhu{{\hat{\bf x}}}
\def\qh{\hat{q}}
\def\uh{\hat{u}}
\def\tauh{\hat{\tau}}

\def\xt{\tilde{x}}
\def\xb{\bar{x}}
\def\cA{{\cal A}}
\def\cG{{\cal G}}
\def\cF{{\cal F}}
\def\cL{{\cal L}}
\def\cC{{\cal C}}
\def\cS{{\cal S}}
\def\cD{{\cal D}}
\def\cK{{\cal K}}
\def\cU{{\cal U}}
\def\cN{{\cal N}}
\def\cI{{\cal I}}
\def\bR{\mathbb R}
\def\cK{\mathcal K}
\def\bZ{\mathbb Z}
\def\cKL{\cal KL}
\def\cCh{\hat{{\cal C}}}
\begin{frontmatter}

\title{Razumikhin-type ISS Lyapunov function and small gain theorem for discrete time time-delay systems with application to a biased min-consensus protocol}

\thanks[footnoteinfo]{
	This work was supported by the National Natural Science Foundation of China under Grant No. 62303112 and Grant No. 62203109, and Natural Science Foundation of Jiangsu Province under Grant No. BK20220812.        }
\thanks[ack2]{Corresponding author: Wenwu Yu.}
\author[First]{Yuanqiu Mo}\ead{yuanqiumo@seu.edu.cn},
\author[First,Second]{Wenwu Yu\thanksref{ack2}}\ead{wwyu@seu.edu.cn},
\author[First,Second]{Huazhou Hou}\ead{huazhouhou@gmail.com},
\author[Third]{Soura Dasgupta}\ead{soura-dasgupta@uiowa.edu}

\address[First]{Southeast University, Nanjing 211189 China}
\address[Second]{Purple Mountain Laboratories, Nanjing 211111 China}
\address[Third]{University of Iowa, Iowa City (IA) 52242 USA}

\begin{abstract}                          
This paper considers small gain theorems for the global asymptotic and exponential input-to-state stability for discrete time time-delay systems using Razumikhin-type Lyapunov function. Among other things, unlike the existing literature, it provides both necessary and sufficient conditions for exponential input-to-state stability in terms of  the Razumikhin-type Lyapunov function and the small gain theorem. Previous necessary ad sufficient conditions were with the more computationally onerous, Krasovskii-type Lyapunov functions. The  result finds application in the robust stability analysis of a graph-based distributed algorithm, namely, the biased min-consensus protocol, which can be used to compute the length of the shortest path from each node to its nearest source in a graph. We consider the biased min-consensus protocol under perturbations that are common in communication networks, including  noise,  delay and asynchronous communication. By converting such a perturbed protocol into a discrete time time-delay nonlinear system, we prove its exponential input-to-state stability under perturbations using our Razumikhin-type Lyapunov-based small gain theorem. Simulations are provided to verify  the theoretical results.
\end{abstract}

\end{frontmatter}

\section{Introduction}
We provide necessary and sufficient conditions for the exponential input-to-state stability (expISS), defined in the sequel, of systems with delay using Razumikhin-type Lyapunov functions also defined below. This stands in contrast to existing results which provide only sufficient conditions.  In the last few decades there have been many contributions to the input-to-state stability (ISS) analysis of discrete time nonlinear systems. Roughly speaking, a system is ISS (resp. exponentially input-to-state stable, expISS) if its state trajectory with bounded
input remains bounded, and asymptotically (resp. exponentially) drops below a
function that increases with  the magnitude of the input. Among various stability tools used, ISS Lyapunov functions and Lyapunov-based small gain theorems have received considerable attention. Inspired by the ISS Lyapunov function established in \cite{jiang2001input}, ISS Lyapunov functions have been extensively studied in  discrete time. In \cite{geiselhart2017equivalent}, three types of ISS Lyapunov functions, namely, max-form, implication-form and dissipative-form ISS Lyapunov functions, are proposed, and their equivalence is characterized. 
\cite{geiselhart2016relaxed} introduces the dissipative-form finite-step  ISS Lyapunov function, which modifies the classical Lyapunov function to permit decrease in a fixed finite number of steps rather than at each step and incorporates that in \cite{jiang2001input} as a special one-step case. Further, the corresponding other finite-step versions of ISS Lyapunov functions proposed in \cite{noroozi2017nonconservative}, as well as their equivalence are characterized. With the proliferation of large-scale systems, Lyapunov-based small gain theorems for interconnected discrete time systems have also been well established. Leveraging the small gain condition designed in \cite{ruffer2010monotone}, small gain theorems in terms of max-form and dissipative-form finite-step ISS Lyapunov functions are established in \cite{geiselhart2016relaxed}. 
Using  max-form finite-step ISS Lyapunov functions, \cite{noroozi2017nonconservative} gives  sufficient and necessary small gain condition for the ISS of interconnected discrete time systems.

Though not as many as those for regular discrete time systems, there are also several  Lyapunov-based results for the ISS of discrete time time-delay systems.  There are two types of ISS Lyapunov functions that are mainly used for systems with delays:
Krasovskii-type and Razumikhin-type ISS Lyapunov functions. The former requires the construction of a
Lyapunov functional making use of an
augmentation of the state vector with all delayed states, while the latter  relies on a Lyapunov-type function defined in the original, non-augmented state space. Consequently, the Krasovskii approach, being constructed in a higher dimensional space, is computationally more complex than  Razumikhin\cite{gielen2012input}.

 In \cite{LIU2009567}, max-form and dissipative-form Razumikhin-type ISS Lyapunov functions have been developed for the ISS and expISS of discrete time time-delay systems. In \cite{gielen2012input}, sufficient conditions for  ISS in terms of Krasovskii-type and Razumikhin-type ISS Lyapunov functions have been derived, and it has been shown that the Krasovskii-type ISS Lyapunov function can be constructed using its  Razumikhin-type counterpart. Sufficient and necessary conditions for the ISS of discrete time time-delay nonlinear systems in terms of Krasovskii-type ISS Lyapunov functions have been characterized in \cite{pepe2017lyapunov}, and such conditions are further derived for discrete time delay-dependent nonlinear systems \cite{pepe2020}. Though, in general, the Razumikhin method is known to provide  only sufficient conditions, \cite{gielen2013necessary} has provided Razumikhin-type sufficient and necessary conditions for the semi-global asymptotic stability and global exponential stability, as opposed to expISS or ISS of delay difference equations. 
 
 While there exist several results on Lyapunov-based small gain theorems for interconnected discrete time systems, very few papers study small gain approaches to the stability analysis of interconnected discrete time time-delay systems. These are needed for  networked control systems like multiagent systems \cite{xu2018consensusability}, formation control \cite{jia2021time} where frequent transmission delays are manifest. The seminal work \cite{gielen2012input} gives a small-gain condition for the ISS of discrete time nonlinear systems with local delays using both max-form Razumikhin-type and Krasovskii-type ISS Lyapunov functions. The Krasovskii-type Lyapunov-based small gain theorems for the global asymptotic stability and ISS of discrete time time-delay systems are in \cite{battista2018small}.

A key point motivating this paper is  the lack of necessary and sufficient conditions for expISS of delay systems using Razumikhin as opposed to Karsovskii-type ISS Lyapunov functions. As noted above, use of Razumikhin-type ISS Lyapunov functions is preferable as they are computationally simpler than their Krasovskii counterparts. Thus, we develop  dissipative-form, as opposed to max-form used in \cite{gielen2012input}, Razumikhin-type ISS Lyapunov functions and small gain conditions  for the ISS and expISS of discrete time time-delay system. The derived  results are non-conservative as they are also  necessary for expISS. 

It is important to note the difference between max and dissipative form Lyapunov functions. Max-form ones  also satisfy the requirement of their dissipative counterparts, though the converse does not always hold. Thus the class of  systems admitting dissipative-form Lyapunov functions is broader than those having max-form ones. Thus using dissipative form Lyapunov functions is much more desirable  as they are more widely applicable.

We further apply our stability result to the robust stability analysis of a biased min-consensus protocol \cite{zhang2017perturbing,mo2019robustness}, which provides a distributed solution to the shortest path finding problem. Previous papers only analyzed its behavior under separate single type of perturbations, e.g., \cite{mo2019robustness} proved its ultimate boundedness under additive bounded noise, and \cite{zhang2017perturbing} studied its convergence under time delays and asynchronous communication. By leveraging the  small gain approach of this paper,  we prove that the biased min-consensus protocol is expISS under the simultaneous presence of  persistent noise, time delays and asynchronous communication, i.e., the estimation error of the protocol will decrease exponentially fast below a bound determined by the extent of these perturbations. 

The rest of the paper is organized as follows: Section \ref{sec:notation} introduces  notations and definitions. Section \ref{sec:main} proposes the Razumikhin-type ISS Lyapunov function and the Lyapunov-based small gain theorem. Section \ref{sec:app} demonstrates the efficacy of the proposed ISS Lyapunov function by applying it to the stability analysis of the biased min-consensus protocol. Section \ref{sec:simulations} provides the simulation. Section \ref{sec:conclusion} concludes.

\subsection{Notations and Definitions}\label{sec:notation}
Define $\mathbb{R}, \mathbb{R}_+, \mathbb{Z}$ and $\mathbb{Z}_+$ as the set of real numbers, the set of nonnegative real numbers, the set of integers and the set of nonnegative integers, respectively. For $(c_1,c_2) \in \mathbb{R}^2$ with $c_1 < c_2$ and $\Pi \subseteq \mathbb{R}$, define $\Pi_{\geq c_1} := \{k \in \Pi ~|~ k \geq c_1 \}$ and $\Pi_{[c_1, c_2]} := \{k \in \Pi ~|~ c_1 \leq k \leq c_2\}$. For any $x \in \mathbb{R}^n$, denote $|x|$, $|x|_\infty$ and $||x||$ as the Euclidean norm, the $\ell_\infty$ norm and any arbitrary monotonic norm on $x$, respectively. For any function $\phi: \mathbb{Z}_+ \rightarrow \bR^m$, its sup-norm is denoted by $||\phi||_{\infty} = \sup\{||\phi(k)||: k \in \bZ_{+}\}$. The set of all functions $\mathbb{Z}_+ \rightarrow \mathbb{R}^m$ with finite sup-norm is denoted by $\ell^\infty$. Let $x := \{x(l) \in \mathbb{R}^n \}_{l \in \mathbb{Z}}$ denote an arbitrary sequence, define $x_{[c_1,c_2]} := \{x(l) \}_{l \in \mathbb{Z}_{[c_1,c_2]}}$ with $(c_1, c_2)\in \mathbb{Z}^2$ and $c_1 < c_2$ as a sequence ordered monotonically with respect to the index $l \in \mathbb{Z}_{[c_1,c_2]}$. With a slight abuse of notation, $||x_{[c_1,c_2]}|| := \max_{l \in [c_1, c_2]}\{||x(l)|| \}$. Further, $\mathrm{id}: \mathbb{R}_+^n \rightarrow \mathbb{R}_+^n$ denotes the identity function on $\mathbb{R}_+^n$, i.e., $\mathrm{id}(x) = x, \forall x \in \mathbb{R}_+^n$. We use $\lambda_1\circ\lambda_2$ to denote the composition of two functions $\lambda_1: \bR^n \rightarrow \bR^n$ and $\lambda_2: \bR^n \rightarrow \bR^n$. Further, we use $\underset{i=1}{\overset{n}{\mathrm{C}}} \lambda_i$ to denote the composition of $\lambda_i: \bR^n \rightarrow \bR^n$ with $i \in \{1,2,\cdots,n\}.$ A function $\alpha: \bR_+ \rightarrow \mathbb{R}_+$ is said to belong to class $\mathcal{K}$ if it is continuous, strictly increasing and $\alpha(0) = 0$. Moreover, $\alpha \in \mathcal{K}_{\infty}$ if $\alpha \in \mathcal{K}$ and $\mathrm{lim}_{s\rightarrow\infty}\alpha(s) = \infty$.

A function $\beta$ is said to belong to class $\mathcal{KL}$ if for a fixed $s \in \mathbb{R}_+$, $\beta(\cdot,s) \in \mathcal{K}$, and for a fixed $r \in \mathbb{R}_+$, $\beta(r,\cdot)$ is decreasing and $\mathrm{lim}_{s \rightarrow \infty}\beta(\cdot,s) = 0$. For $\alpha \in \mathcal{K}$, we write $\alpha < \mathrm{id}$ to mean $\alpha(s) < s$ for all $s \neq 0$.

We consider $\ell$ interconnected discrete time time-delay nonlinear subsystems such that the $i$-th subsystem obeys
\begin{flalign}\label{eq:composite}
x_i(k\!+\!1) \!=\!  f_i\big(x_1(k \!-\! \tau_{i1}(k)),\! \cdots\!, x_\ell(k \!-\! \tau_{i\ell}(k)), u(k \!-\! d) \big)
\end{flalign}
where $k \in \mathbb{Z}_+$, $f_i: \mathbb{R}^{n_1} \times \cdots \times \mathbb{R}^{n_\ell}\times\mathbb{R}^m \rightarrow \mathbb{R}^{n_i}$ may not be continuous and satisfies $f_i(0,\cdots,0) = 0$, $\tau \geq \tau_{ij}(k) \in \mathbb{Z}_+$ with $j \in \{1,\cdots, \ell\}$ reflects the time delay between subsystem $i$ and subsystem $j$, with $\tau$ denoting the maximum time delay and $\tau_{ij}(k) = 0$ indicating that there is no time delay between $i$ and $j$ at time $k$, and $d \in \mathbb{Z}_+$ denotes the delay in the input. 

Let $x(k) := [x_1(k)^\mathrm{T} ~ \cdots~ x_\ell(k)^\mathrm{T}]^\top \in \mathbb{R}^n$ with $n = \sum_{i=1}^{\ell} n_i$, the composite system can be described by 
\begin{equation}\label{eq:compositesystem}
x(k+1) = G\big(x_{[k-\tau, k]}, u(k-d)   \big), ~ k \in \mathbb{Z}_+
\end{equation}
where $x_{[k-\tau, k]} \in (\mathbb{R}^n)^{\tau + 1}$, $G: (\mathbb{R}^n)^{\tau + 1} \times \mathbb{R}^m \rightarrow \mathbb{R}^n$ and $G(0_{[k-\tau, k]}, 0) = 0$ 

We use $\{x(k, \xi_{[-\tau,0]}, u_{[0,k-1-d]})  \}_{k \in \mathbb{Z}_{\geq 1}}$ to denote the trajectory of the system (\ref{eq:compositesystem}) with $\xi_{[-\tau,0]} \in (\mathbb{R}^n)^{\tau + 1}$ the initial state and $u_{[0,k-1-d]} := \{u(l)\}_{l \in \mathbb{Z}_{[0,k-1-d]}}, u(l) \in \mathbb{R}^m$ the input. Similarly, $\{x_i(k, \xi_{[-\tau,0]}, u_{[0,k-1-d]})  \}_{k \in \mathbb{Z}_{\geq 1}}$ is used to denote the trajectory of (\ref{eq:composite}). To simplify the notation, we further denote $x(k) := x(k, \xi_{[-\tau,0]}, u_{[0,k-1-d]})$ and $x_i(k) := x_i(k, \xi_{[-\tau,0]}, u_{[0,k-1-d]})$ for $k \in \bZ_{\geq 1}$. Note that by the equivalence of norms, for any norm $||\cdot||$ on $\mathbb{R}^n$, there exists a constant $q \geq 1$ such that
\begin{equation}\label{eq:q}
    || x(k) || \leq q \max_{i \in \{1,\cdots,\ell \} } ||x_i(k)||
\end{equation}

\begin{definition}\label{def:expISS}
\cite{LIU2009567} We call (\ref{eq:compositesystem}) globally asymptotically input-to-state stable (ISS) if there exist $\beta \in \mathcal{KL}$ and $\lambda \in \mathcal{K}$ such that for all initial states $\xi_{[-\tau,0]} \in (\mathbb{R}^n)^{\tau + 1}$, all inputs $u(\cdot) \in \ell^\infty(\mathbb{R}^m)$ and all $k \in 
	\bZ_{+}$
	\begin{equation}\label{eq:expISS}
	||x(k)|| \leq \beta(||\xi_{[-\tau,0]}||,k) + \lambda(||u||_\infty),
	\end{equation}
	In particular, following the definition in \cite{geiselhart2016relaxed}, if $\beta$ in (\ref{eq:expISS}) can be chosen as
	\begin{equation}\label{eq:rate}
	\beta(r,k) = p\rho^kr
	\end{equation}
	with $p \geq 1$ and $\rho \in [0,1)$, then (\ref{eq:composite}) is called globally exponentially input-to-state stable (expISS).  
\end{definition}
Though $G$ in (\ref{eq:compositesystem}) is not required to be continuous, it needs to satisfies the $\mathcal{K}$-boundedness property introduced in the following definition throughout the paper.
\begin{definition}\label{def:kbound}
The function $G$ in (\ref{eq:composite}) is globally $\mathcal{K}$-bounded, i.e., there exist functions $\omega_1,\omega_2 \in \mathcal{K}$ such that for all $\xi = \{\xi(l) \}_{l \in [1, \tau + 1]} \in (\mathbb{R}^{n})^{\tau + 1}$ and $\mu \in \mathbb{R}^{m}$ such that
\begin{equation}\label{eq:kbound}
||G(\xi, \mu)|| \leq \omega_1(||\xi||) + \omega_2(||\mu||).
\end{equation}
\end{definition}
\begin{remark}\label{re:continuity}
It follows directly from (\ref{eq:kbound}) that global $\mathcal{K}$-boundedness implies continuity of $G$ at the origin and boundedness of $G$ on bounded sets. The converse implication also holds true by Lemma 5 in \cite{geiselhart2017equivalent}. Further, it follows from Remark 3.3 in \cite{geiselhart2016relaxed} that global $\mathcal{K}$-boundedness is a necessary condition for the ISS of discrete time system without time delays, i.e., $x(k+1) = G(x(k), u(k))$ with $G: \mathbb{R}^n \times \mathbb{R}^m \rightarrow \mathbb{R}^n$, and it can be readily verified that such a result can be extended for discrete time time-delay system defined in (\ref{eq:compositesystem}).
\end{remark} 
To derive sufficient and necessary conditions for the expISS of (\ref{eq:compositesystem}), we also assume (\ref{eq:compositesystem}) admits a solution of length $M + 1$ with $M \geq \tau$, per the definition:
\begin{definition}\label{def:length}
    \cite{gielen2013necessary} (\ref{eq:compositesystem}) admits a solution of length $M + 1$ with $M \geq \tau$ if, for each $M \geq \tau$ there holds $x(k + 1) = G(x_{[k - \tau,k]}, u(k - d))$ for all $k \in \mathbb{Z}_{[-M + \tau, 0]}$. Obviously, (\ref{eq:compositesystem}) admits a solution of length $\tau + 1$.
\end{definition}

We further make the following assumptions on each subsystem.
\begin{assumption}\label{ass:vslf}
Consider the subsystem $i \in \{1, \cdots, \ell\}$ defined in (\ref{eq:composite}). There exists a real valued function $V_i: \bR^{n_i} \rightarrow \bR_{+}$ such that the following holds:
\begin{itemize}
	\item There exist $\mathcal{K}_{\infty}$ functions $\alpha_{i1}$ and $\alpha_{i2}$ such that
	\begin{equation}\label{eq:fir}
	\alpha_{i1}(||\xi_i||) \leq V_i(\xi_i) \leq \alpha_{i2}(||\xi_i||),~ \forall \xi_i \in \bR^{n_i}.
	\end{equation}
	\item  With $\tau$ the maximum time delay, there exist linear $\mathcal{K}_\infty$ function $\lambda_{ij}$, $\mathcal{K}$ function $\lambda_{iu}$ and a non-negative integer $M \geq \tau$ such that for all $x_{i_{[-M, 0]}} \in (\bR^{n_i})^{M + 1}$ with $i \in \{1,\cdots, \ell\}$, all $u \in \ell^\infty(\mathbb{R}^m)$ and $k \in \mathbb{Z}_+$, there holds
	\begin{flalign}\label{eq:vslf}
	V_i(x_i(k + 1)) &\leq \max_{\theta \in \mathbb{Z}_{[k-M,k]}, j\in \{1,\cdots, \ell \}}\lambda_{ij}\big(V_j( x_j(\theta)) \big) \nonumber \\
	&+ \lambda_{iu}(||u||_\infty).
	\end{flalign}
\end{itemize}
\end{assumption}
Assumption \ref{ass:vslf} constrains   the trajectory of a subsystem $i$ by   another $j$ and the input at  time  $k $. As such  (\ref{eq:vslf}) is a dissipative Lyapunov inequality. This contrasts with the max-form variation  given in
\cite{gielen2012input}, where (\ref{eq:vslf}) is replaced by $V_i\big(x_i(k + 1)\big) \leq$
\[ 
\max\big\{ \max_{\theta \in \mathbb{Z}_{[k - M,k]},j\in \{1,\cdots, \ell \}}\lambda_{ij}\big(V_j( x_j(\theta) ) \big), \lambda_{iu}(||u||_\infty)\big\}. \]
Here the summation in  (\ref{eq:vslf}) is replaced by a max operation. This max-form inequality  implies (\ref{eq:vslf}) though the converse may not hold highlighting  the wider applicability of dissipative-form Lyapunov functions.

Observe that (\ref{eq:vslf}) puts us in a  finite-step ISS Lyapunov framework, with the maximizing $\theta$ representing the step size over which the constraining inequality holds. Unlike, 
\cite{geiselhart2016relaxed} and \cite{noroozi2017nonconservative} that adopt a similar characterization to construct the finite-step ISS Lyapunov function via a small gain approach, the  $\theta$ in (\ref{eq:vslf})  is allowed to be time-varying. In \cite{geiselhart2016relaxed} and \cite{noroozi2017nonconservative} it is constant over all $k.$ It is important to note that (\ref{eq:fir})-(\ref{eq:vslf}) is not equivalent to the
Razumikhin-type ISS Lyapunov function given in \cite{LIU2009567} as the latter requires $\lambda_{ij}$ in (\ref{eq:vslf}) to satisfy $\lambda_{ij} \leq \rho\mathrm{id}$ with $\rho \in [0, 1)$, which ensures that with zero input the Lyapunov function strictly decreases from its maximum value among previous $M$ steps.

We end this section with the following definition.

\begin{definition}\label{def:csybsys}
We call $\lambda_{ij}$ and $j$ in (\ref{eq:vslf}) the constraining comparison function and the constraining subsystem of subsystem $i$ at time $k + 1$, respectively.
\end{definition}

\section{ISS Lyapunov function and Small gain theorem}\label{sec:main}
In this section, we prove the ISS and expISS for discrete time time-delay system (\ref{eq:compositesystem}) via both Razumikhin-type ISS Lyapunov function and small gain theorem. Furthermore, converse Lyapunov  and small gain theorems for expISS are also given.

\subsection{The Razumikhin-type ISS Lyapunov function}
We first characterize Razumikhin-type Lyapunov based sufficient conditions under which (\ref{eq:compositesystem}) is ISS while ignoring the interconnections. The  Razumikhin-type Lyapunov is in (\ref{eq:uplower}) and (\ref{eq:Lya}) of the theorem statement.
\begin{theorem}\label{the:suffISS}
	The system in (\ref{eq:compositesystem}) is ISS if there exists a function $V: \bR^{n} \rightarrow \bR_{+}$ obeying the following condition:
	
		a) There exist $\mathcal{K}_\infty$ functions $\underline{\alpha}$ and $\bar{\alpha}$ such that
		\begin{equation}\label{eq:uplower}
		\underline{\alpha}(||\xi||) \leq V(\xi) \leq \bar{\alpha}(||\xi||),~ \forall \xi \in \mathbb{R}^n.
		\end{equation}
		b) With $\tau$ the maximum time delay, for all $k \in \mathbb{Z}_+$ there exist $\kappa \in [0, 1)$, $M \in \mathbb{Z}_{\geq \tau}$ and $\lambda_u \in \mathcal{K}$ such that
		\begin{flalign}\label{eq:Lya}
		V(x(k+1)) \leq  \max_{\theta \in \mathbb{Z}_{[k - M, k]}}\kappa V(x(\theta)) + \lambda_u(||u||_\infty)
		\end{flalign}
\end{theorem}
\begin{proof}
From (\ref{eq:Lya}), for all $k \in \mathbb{Z}_+$, $V(x(k + 1))$ obeys
\begin{equation}\label{eq:first}
 V(x(k+1)) \leq \kappa V(x(k - m_1)) + \lambda_u(||u||_\infty)
\end{equation}
where in (\ref{eq:first}) we assume $m_1 = \arg \max_{\theta \in \mathbb{Z}_{[-M, 0]}}\kappa V(x(k - \theta))$. By a simple induction there exist with $m_i \in \mathbb{Z}_{[0, M]}$ such that proceeding in this way, there holds
\begin{flalign}
&V(x(k+1)) \leq  \nonumber \\
& \!\! \lambda_u(||u||_\infty) \!+\! \kappa^2 \big(V(x(k \!-\! m_1\! - \!1 \!- \!m_2))\!\! +\! \!\lambda_u(||u||_\infty) \big)   \label{eq:uag}\\
& \cdots \nonumber \\
&\!\!\leq\! \kappa^s\big(V( x(k \!-\! (s\! -\! 1)\! -\! \! \sum_{i = 1}^{s}m_i))\big)\! +\! \sum_{i = 0}^{s - 1}\kappa^i\lambda_u(||u||_\infty) \label{eq:uu} \\
&\!\! \leq\! \kappa^s\big(V( x(k - (s - 1) -  \sum_{i = 1}^{s}m_i)) \big) + \frac{1}{1 - \kappa}\lambda_u(||u||_\infty) \nonumber 
\end{flalign}
where (\ref{eq:uag}) uses (\ref{eq:Lya}), in (\ref{eq:uu}) we assume $k$ is decomposed as $k = m_1 + \sum_{i = 2}^{s}(m_i + 1) + j$  for some $j \in \mathbb{Z}_{[-\tau, M - \tau]}$, and the last inequality uses $\sum_{i = 0}^{s - 1}\kappa^i < \sum_{i = 0}^{\infty}\kappa^i = \frac{1}{1 - \kappa}$ as $\kappa \in [0, 1)$. As $m_i \leq M$ for all $i \in \{1,\cdots, s\}$, there holds $\sum_{i=1}^{s}m_i = k - j - (s - 1) \leq sM$, leading to $s \geq \frac{k - j + 1}{M + 1}$. 

As $\kappa \in [0, 1)$ and $s \geq \frac{k - j + 1}{M + 1}$, it follows from (\ref{eq:uplower}) that $V(x(k + 1))$ further obeys
\begin{flalign}
&V(x(k + 1)) \leq \kappa^{\frac{k - j + 1}{M + 1}}\big(V( x(j)) \big) + \frac{1}{1 - \kappa}\lambda_u(||u||_\infty) \nonumber \\
&\leq \kappa^{\frac{k + 1 + \tau - M}{M + 1}}\bar{\alpha}(||x(j)||) + \frac{1}{1 - \kappa}\lambda_u(||u||_\infty) \label{eq:jj} \\
&\leq \kappa^{\frac{k + 1 + \tau - M}{M + 1}}\bar{\alpha}(\alpha^\ast(||x_{[-\tau, 0]}||) + \lambda^\ast(|| u ||_\infty)) \nonumber \\
&~~~ + \frac{1}{1 - \kappa}\lambda_u(||u||_\infty) \label{eq:ule1}
\end{flalign}
where (\ref{eq:jj}) uses $j \in \mathbb{Z}_{[-\tau, M - \tau]}$, (\ref{eq:ule1}) uses the fact that $||x_{[0,M-\tau]}|| \leq \alpha^\ast(||x_{[-\tau, 0]}||) + \lambda^\ast(|| u ||_\infty)$ with $\mathcal{K}$ functions $\alpha^\ast \geq \mathrm{id}$ and $\lambda^\ast$  due to the global $\mathcal{K}$-boundedness of $G$ (see Lemma 12 in \cite{bobiti2014input}).
From (\ref{eq:uplower}), $||x(k + 1)||$ obeys
\begin{flalign}
&||x(k + 1)|| \nonumber\\
& \leq \underline{\alpha}^{-1}\circ 2\kappa^{\frac{k + 1}{M + 1}}\kappa^{\frac{\tau - M}{M + 1}}\mathrm{id}\circ\bar{\alpha}(\alpha^\ast(|| x_{[-\tau, 0]} || ) + \lambda^\ast(|| u ||_\infty) ) \nonumber \\
&~~~ +  \underline{\alpha}^{-1}\circ\frac{2}{1 - \kappa}\lambda_u(||u||_\infty) \label{eq:le} \\
&\leq \underline{\alpha}^{-1}\circ 2\kappa^{\frac{k + 1}{M + 1}}\kappa^{\frac{\tau - M}{M + 1}}\mathrm{id}\circ\bar{\alpha}\circ2\mathrm{id}\circ\alpha^\ast(|| x_{[-\tau, 0]} || )\nonumber \\
&~~~ + \bar{\lambda}_u(||u||_\infty) \label{eq:sl}
\end{flalign}
where (\ref{eq:le}) and (\ref{eq:sl}) use the fact that $\alpha(a + b) \leq \alpha(2a) + \alpha(2b)$ for any $\mathcal{K}$ function $\alpha$ and all $a, b \in \mathbb{R}_+$, and in (\ref{eq:sl}) $\bar{\lambda}_u= \underline{\alpha}^{-1}\circ 2\kappa^{\frac{1 + \tau - M}{M + 1}}\mathrm{id}\circ\bar{\alpha}\circ2\mathrm{id}\circ\lambda^\ast + \underline{\alpha}^{-1}\circ\frac{2}{1 -\kappa}\mathrm{id}\circ\lambda_u \in \mathcal{K}$.
Let $c = 2\kappa^{\frac{\tau - M}{M + 1}} \geq 1$, $\bar{\kappa} = \kappa^{\frac{1}{M + 1}} \in [0, 1)$, and $\beta(r, t) :=  \underline{\alpha}^{-1}\circ c\bar{\kappa}^t\mathrm{id}\circ\bar{\alpha}\circ\alpha^\ast(r)$. Obviously, $\beta \in \mathcal{KL}$, and (\ref{eq:sl}) becomes 
\begin{equation}\label{eq:exp}
||x(k + 1)|| \leq \beta(|| x_{[-\tau, 0]} ||, k + 1) + \bar{\lambda}_u(||u||_\infty) 
\end{equation}
Therefore, (\ref{eq:compositesystem}) is ISS by Definition \ref{def:expISS}. 
\end{proof}
Indeed the $V(\cdot)$ in this theorem is a  dissipative-form Razumikhin-type ISS Lyapunov function.
\begin{definition}\label{def:Razumikhin}
The real valued function $V: \bR^{n} \rightarrow \bR_{+}$ satisfying (\ref{eq:uplower}) and (\ref{eq:Lya}) is called a dissipative-form Razumikhin-type ISS Lyapunov function for (\ref{eq:compositesystem}).
\end{definition}
It can be observed from (\ref{eq:uplower})-(\ref{eq:Lya}) that the constraint for each subsystem given by (\ref{eq:fir})-(\ref{eq:vslf}) in Assumption \ref{ass:vslf} is not the dissipative-form Razumikhin-type ISS Lyapunov function defined in Definition \ref{def:Razumikhin} as $\lambda_{ij}$ in (\ref{eq:vslf}) is not required to satisfy $\lambda_{ij} \leq \kappa\mathrm{id}$ with $\kappa \in [0,1)$.

We now turn to a sufficient condition for expISS.
\begin{corollary}\label{corr:exp}
	Suppose conditions in Theorem \ref{the:suffISS} hold,  $\underline{\alpha}, \bar{\alpha}$ in (\ref{eq:uplower}) and $\omega_1$ in (\ref{eq:kbound}) are linear. Then (\ref{eq:compositesystem}) is expISS. 
\end{corollary}
\begin{proof}
As $\underline{\alpha}$ in (\ref{eq:uplower}) is linear, by (\ref{eq:uplower}) and (\ref{eq:ule1}),
\begin{flalign}
&\!\!\!\!||x(k + 1)|| \leq  \underline{\alpha}^{-1}\circ\frac{1}{1 - \kappa}\lambda_u(||u||_\infty) + \nonumber \\
&\!\!\!\! \underline{\alpha}^{-1}\circ c\bar{\kappa}^{k+1}\mathrm{id}\circ\bar{\alpha}(\alpha^\ast(||x_{[-\tau, 0]}||) + \lambda^\ast(|| u ||_\infty)) \label{eq:lin} \\
&\!\!\!\!= \underline{\alpha}^{-1}\circ c\bar{\kappa}^{k+1}\mathrm{id}\circ\bar{\alpha}\circ \alpha^\ast(||x_{[-\tau, 0]}||) + \bar{\lambda}_u(||u||_\infty) \label{eq:an}
\end{flalign}
where in (\ref{eq:lin}) $c = \kappa^{\frac{\tau - M}{M + 1}} \geq 1$, $\bar{\kappa} = \kappa^{\frac{1}{M+1}} \in [0, 1)$ with $\kappa$ defined in (\ref{eq:Lya}), and in (\ref{eq:an}) $\bar{\lambda}_u = \underline{\alpha}^{-1}(\bar{\kappa}^{\frac{1 + \tau - M}{M + 1}}\mathrm{id}\circ\bar{\alpha}\circ \lambda^\ast + \frac{1}{1 - \kappa}\lambda_u) \in \mathcal{K}$ with $\lambda_u$ defined in (\ref{eq:Lya}). 

As $\omega_1$ is linear, $\alpha^\ast$ is linear (see Corollary 5.7 in \cite{geiselhart2016relaxed}). Let $\beta(r,t) = \underline{\alpha}^{-1}\circ c\bar{\kappa}^t\mathrm{id}\circ\bar{\alpha}\circ\alpha^\ast(r)$. $\beta(r,t)$ is linear for a fixed $t$. Based on Definition \ref{def:expISS}, (\ref{eq:compositesystem}) is expISS, e.g., in this case by choosing $\rho = \bar{\kappa} \in [0, 1)$, we can always find a $\beta'(r, t) = p\rho^tr \geq \beta(r,t)$ for all $r \in \mathbb{R}_+$ and all $t \in \mathbb{Z}_+$ with $p \geq 1$. 
\end{proof}
While Corollary \ref{corr:exp} provides Razumikhin-type sufficient conditions for expISS, the following theorem further proves that such conditions are also necessary for the expISS of (\ref{eq:compositesystem}).
\begin{theorem}\label{the:conLya}
	The system (\ref{eq:compositesystem}) is expISS if and only if it admits the Razumikhin-type ISS Lyapunov defined in Definition \ref{def:Razumikhin}. Moreover, $\bar{\alpha}$ and $\underline{\alpha}$ defined in (\ref{eq:uplower}), as well as $\omega_1$ defined in (\ref{eq:kbound}) are linear.
\end{theorem}
\begin{proof}
	The sufficiency is proved in Corollary \ref{corr:exp}. For  necessity, as (\ref{eq:compositesystem}) is expISS, it follows from Definition \ref{def:expISS} that 
	\begin{equation}\label{eq:to}
	||x(k)|| \leq p\rho^k(||\xi_{[-\tau,0]}||) + \lambda(||u||_\infty)
	\end{equation}
	with $p \geq 1$, $\rho \in [0, 1)$ and $||\xi_{[-\tau,0]}|| \in (\mathbb{R}^n)^{\tau + 1}$ the initial state. Setting $k = 1$ in (\ref{eq:to}), we obtain 
	\begin{eqnarray}
	||x(1)|| & = & G(\xi_{[-\tau,0]}, u(0 - d)) \nonumber \\
	&\leq& p\rho(||\xi_{[-\tau,0]}||) + \lambda(||u||_\infty), \nonumber
	\end{eqnarray}
	and thus (\ref{eq:kbound}) holds with a linear $\mathcal{K}$ function $\omega_1 = p\rho\mathrm{id}$.
	
	Further, it follows from (\ref{eq:to}) that for any $k + 1 \geq \bar{M} > \log_\rho \frac{1}{p} > 0$ with $\bar{M} \in \mathbb{Z}_{\geq 1}$, there exists a $\bar{\rho} \in [0, 1)$ such that for $k \in \mathbb{Z}_{\geq \bar{M} - 1}$ 
	\begin{equation}\label{eq:kb}
	||x(k + 1)|| \leq \bar{\rho}(||\xi_{[k + 1 -\tau - \bar{M}, k + 1 - \bar{M}]}||) + \lambda(||u||_\infty). 
	\end{equation}
	Let $V(\cdot) := ||\cdot||$. Then (\ref{eq:uplower}) holds with $\bar{\alpha} = \underline{\alpha} = \mathrm{id}$. As (\ref{eq:compositesystem}) admits a solution of length $M + 1$ with $M \geq \tau$, let $M = \tau + \bar{M} - 1 \in \mathbb{Z}_{\geq \tau}$, it follows from (\ref{eq:kb}) that for all $k \in \mathbb{Z}_+$
	\begin{equation}
	V(x(k+1)) \leq  \max_{\theta \in \mathbb{Z}_{[k- M, k]}}\bar{\rho} V(x(\theta)) + \lambda(||u||_\infty),
	\end{equation}  
	completing our proof.
\end{proof}

\subsection{The small gain theorem for interconnected systems}
Having provided a necessary and sufficient condition for the expISS of the delay difference equations using a Razumikhin-type ISS Lyapunov function, we now turn to providing  a  small gain condition for interconnected systems involving  delay difference equations. The theorem explicitly takes into account the properties of system interconnections. We need the following assumption:
\begin{assumption}\label{ass:small}
The linear $\mathcal{K}_\infty$ functions $\lambda_{ij}$ in (\ref{eq:vslf}) satisfy
\begin{equation}\label{eq:small}
    \lambda_{i_1i_2}\circ\lambda_{i_2i_3}\circ\cdots \circ \lambda_{i_{r-1}i_r} < \mathrm{id}
\end{equation}
for all sequences $(i_1, \cdots, i_r) \in \{1, \cdots, \ell \}^r$ with $r \in \{1, \cdots, \ell\}$.
\end{assumption}
This small gain condition  in Assumption \ref{ass:small} follows those in \cite{noroozi2017nonconservative,ruffer2010monotone,geiselhart2016relaxed}, the main difference between this work and these others is in two aspects: 1) in this work the time step (i.e., $\theta$ in (\ref{eq:vslf})) in Lyapunov-like state estimate defined in Assumption \ref{ass:vslf} for each subsystem can be time-varying while theirs relies on a fixed time dependency; and 2) by utilizing the above small gain condition this work aims to establish a dissipative-form Razumikhin-type ISS Lyapunov function defined in Definition \ref{def:Razumikhin}, allowing the maximizing $\theta$ in (\ref{eq:Lya}) to be time-varying, while theirs uses a finite-step ISS Lyapunov framework.

The Lyapunov-based small gain theorem for the ISS of (\ref{eq:compositesystem}) is given below.
\begin{theorem}\label{thm:first}
Suppose Assumptions \ref{ass:vslf} and \ref{ass:small} hold. Then (\ref{eq:compositesystem}) admits a dissipative-form Razumikhin-type ISS Lyapunov function given in (\ref{eq:uplower}) and (\ref{eq:Lya}) and is ISS.
\end{theorem}
\begin{proof}
As the $\mathcal{K}_\infty$ functions $\lambda_{ij}$ in (\ref{eq:Lya}) is linear, it follows \cite{geiselhart2016relaxed} (see Corollary 5.7 and Theorem 6.4) and \cite{geiselhart2012numerical} that there exist linear $\mathcal{K}_\infty$ functions $\sigma_i$ with $i \in \{1,\cdots, \ell \}$ and $\kappa \in [0, 1)$  such that 
\begin{equation}\label{eq:smallr}
\max_{j \in \{1,\cdots, \ell \}} \sigma_i^{-1}\circ\lambda_{ij}\circ\sigma_j < \kappa\mathrm{id}.
\end{equation}
Define a real valued function $V: \mathbb{R}^n \rightarrow \mathbb{R}_+$ as
\begin{equation}\label{eq:Lya1}
V(\xi) := \max_{i\in \{1, \cdots,\ell\}} \sigma_i^{-1}(V_i(\xi)) 
\end{equation}
with $\xi = [\xi_1^\top, \cdots, \xi_\ell^\top]^\top$ and $V_i$ defined in Assumption \ref{ass:vslf}. From (\ref{eq:fir}), there exist $\mathcal{K}_\infty$ functions $\underline{\alpha}$ and $\bar{\alpha}$ such that
\begin{equation}\label{eq:lu}
\underline{\alpha}(||\xi||) \leq V(\xi) \leq \bar{\alpha}(||\xi||),~ \forall \xi \in \mathbb{R}^n.
\end{equation} 
Let $i = \arg \max_{i \in \{1, \cdots,\ell\}} \sigma_i^{-1}(V_i(x_i(k + 1)))$ for some $k \in \mathbb{Z}_+$. Then it follows from (\ref{eq:vslf}) that
\begin{flalign}
& V(x(k+1)) = \sigma_i^{-1}(V_i(x_i(k + 1)))   \nonumber \\
&  \leq \sigma_i^{-1}\big(\max_{\substack{\theta \in \mathbb{Z}_{[k-M,k]} \\ j\in \{1,\cdots, \ell \}}}\lambda_{ij}\big(V_j( x_j(\theta) ) \big)   + \lambda_{iu}(||u||_\infty)\big)  \label{eq:uas} \\
& =  \sigma_i^{-1}\big(\max_{\theta \in \mathbb{Z}_{[k-M,k]}, j\in \{1,\cdots, \ell \}}\lambda_{ij}\circ\sigma_j\circ\sigma_j^{-1}\circ\big(V_j( x_j(\theta) ) \big) 
\nonumber \\
&~~~ + \lambda_{iu}(||u||_\infty)\big)\nonumber \\  
&\leq  \sigma_i^{-1}\big(\max_{\substack{\theta \in \mathbb{Z}_{[k-M,k]} \\ j,l\in \{1,\cdots, \ell \}}}\lambda_{ij}\circ\sigma_j\circ\sigma_l^{-1}\circ\big(V_l( x_l(\theta) ) \big)\big) \nonumber \\
&~~~ + \sigma_i^{-1}\circ\lambda_{iu}(||u||_\infty) 
\label{eq:nad} \\
& \leq  \max_{\theta \in \mathbb{Z}_{[k-M, k]}}\kappa V(x(\theta)) + \lambda_u(||u||_\infty) \label{eq:fiu}
\end{flalign}
where in (\ref{eq:uas}) we assume $j$ is the constraining subsystem (per Definition \ref{def:csybsys}) of $i$ at time $k + 1$, (\ref{eq:nad}) uses the fact that $\sigma_i^{-1}$ is linear, and (\ref{eq:fiu}) uses (\ref{eq:smallr}) and (\ref{eq:Lya1}). Then it follows from Definition \ref{def:Razumikhin} that $V$ in (\ref{eq:Lya1}) is a dissipative-form Razumikhin-type ISS Lyapunov function for (\ref{eq:compositesystem}), and thus (\ref{eq:compositesystem}) is ISS by  Theorem \ref{the:suffISS}.
\end{proof}
The following Corollary further characterizes the conditions under which (\ref{eq:compositesystem}) is expISS using the Lyapunov-based small gain theorem introduced in Theorem \ref{thm:first}.
\begin{corollary}\label{co:small}
	Suppose conditions in Theorem \ref{thm:first} hold. Furthermore, suppose $\alpha_{i1}$ and $\alpha_{i2}$ in (\ref{eq:fir}) and $\omega_1$ in (\ref{eq:kbound}) are linear. Then (\ref{eq:compositesystem}) is expISS. 
\end{corollary}
\begin{proof}
	As $\alpha_{i1}$, $\alpha_{i2}$ in (\ref{eq:fir}) and $\sigma_i$ in (\ref{eq:smallr}) are linear, it follows from (\ref{eq:Lya1}) that $\bar{\alpha}$ and $\underline{\alpha}$ in (\ref{eq:lu}) can be linear. Further, Theorem \ref{thm:first} implies that (\ref{eq:compositesystem}) admits a Razumikhin-type ISS Lyapunov function if conditions in Theorem \ref{thm:first} hold. By Corollary \ref{corr:exp}, (\ref{eq:compositesystem}) is expISS.
\end{proof}
Now we are ready to give the converse small gain theorem for the expISS of the interconnected system in (\ref{eq:compositesystem}).
\begin{theorem}
Consider the interconnected discrete time time-delay system (\ref{eq:compositesystem}). It is expISS if and only if conditions in Theorem \ref{thm:first} hold and $\omega_1$ in (\ref{eq:kbound}) of Definition \ref{eq:kbound}, as well as $\alpha_{i1}$ and $\alpha_{i2}$ in (\ref{eq:fir}) are linear.
\end{theorem}
\begin{proof}
	The sufficiency is proved in Corollary \ref{co:small}. For the necessity, it follows from Theorem \ref{the:conLya} that expISS of (\ref{eq:compositesystem}) implies that $\omega_1$ in (\ref{eq:kbound}) is linear and there exists a Razumikhin-type ISS Lyapunov defined in Definition \ref{def:Razumikhin} for (\ref{eq:compositesystem}), with $\bar{\alpha}$ and $\underline{\alpha}$ in (\ref{eq:uplower}) both linear. 
	Define $V_i : \mathbb{R}^{n_i}\rightarrow \mathbb{R}_+$ by $V_i(\cdot):= ||\cdot||$ for $i \in \{1, \cdots, \ell \}$. Then (\ref{eq:fir}) is satisfied with $\alpha_{i1} = \alpha_{i2} = \mathrm{id}$. For $k \in \mathbb{Z}_+$, there holds
	\begin{flalign}
	& V_i(x_i(k + 1)) = || x_i(k + 1) || \leq || x(k + 1) ||  \nonumber \\
	& \!\leq\! \underline{\alpha}^{-1}\circ c\bar{\kappa}^{k+1}\mathrm{id}\circ\bar{\alpha}\circ\alpha^\ast(|| x_{[-\tau, 0]}||)\! +\! \bar{\lambda}_u(||u||_\infty) \label{eq:mj} \\
	&\!=\! \underline{\alpha}^{-1}\circ c\bar{\kappa}^{k+1}\mathrm{id}\circ\bar{\alpha}\circ\alpha^\ast(\max_{\theta \in \mathbb{Z}_{[-\tau, 0]}}||x(\theta)||)\! +\! \bar{\lambda}_u(||u||_\infty) \nonumber \\
	& \!\leq\! \underline{\alpha}^{-1}\!\circ\! c\bar{\kappa}^{k+1}\mathrm{id}\!\circ\!\bar{\alpha}\!\circ\!\alpha^\ast \big(q\max_{\theta \in \mathbb{Z}_{[-\tau,0]}}V_j(x_j(\theta))\big) \! +\! \bar{\lambda}_u(||u||_\infty) \label{eq:norm1}
	\end{flalign}
	where (\ref{eq:mj}) uses (\ref{eq:an}) in Corollary \ref{corr:exp}, (\ref{eq:norm1}) uses (\ref{eq:q}) and the fact that $V_j(\cdot):= ||\cdot||$, and in (\ref{eq:norm1}) we assume $j = \arg \max_{i \in \{1,\cdots,\ell \} } \{||x_i(\theta)||\}$. As $\bar{\kappa} \in [0, 1)$, $\underline{\alpha}$ and $\bar{\alpha}$ are both linear, from (\ref{eq:norm1}), there must exist an $\bar{M} \in \mathbb{Z}_+$ such that when $k \in \mathbb{Z}_{\geq \bar{M}}$, $\underline{\alpha}^{-1}\circ c\bar{\kappa}^{k+1}\mathrm{id}\circ\bar{\alpha}\circ \alpha^\ast \circ q\mathrm{id} < \rho\mathrm{id}$ with $\rho \in [0,1)$, yielding that for $k \in \mathbb{Z}_{\geq \bar{M}}$,
	\begin{equation*}
	V_i(x_i(k + 1)) \leq  \rho\max_{\substack{\theta \in \mathbb{Z}_{[k-\tau-\bar{M},k - \bar{M}]} \\ j\in \{1,\cdots, \ell \}}}V_j(x_j(\theta))  + \lambda_{iu}(||u||_\infty)
	\end{equation*}
	with $\lambda_{iu} = \bar{\lambda}_u \in \mathcal{K}$. As (\ref{eq:compositesystem}) admits a solution of length $M + 1$ with $M \geq \tau$, let $M = \bar{M} + \tau$, for all $k \in \mathbb{Z}_+$
	\begin{flalign*}
	V_i(x_i(k + 1)) &\leq \max_{\substack{j \in \{1,\cdots,\ell \}\\ \theta \in \mathbb{Z}_{[k - M, k - M + \tau]}}}\rho V(x_j(\theta)) + \lambda_{iu}(||u||_\infty) \nonumber \\ 
	&\leq \max_{\substack{j \in \{1,\cdots,\ell \}\\ \theta \in \mathbb{Z}_{[k - M, k]}}}\rho V(x_j(\theta)) + \lambda_{iu}(||u||_\infty), \nonumber \\
	\end{flalign*}
	and thus Assumption \ref{ass:vslf} holds with $\lambda_{ij}$ in (\ref{eq:vslf}) obeying $\lambda_{ij} = \rho \mathrm{id}$ for all $i,j \in \{1,\cdots, \ell \}$, which further makes the small gain condition in Assumption \ref{ass:small} hold.
\end{proof}



\section{Applications}\label{sec:app}
In this section, we demonstrate the utility of our small gain theorem by applying it to the robust stability analysis of a biased min-consensus protocol introduced in \cite{zhang2017perturbing}, that computes the  shortest distance from each non-source node to its nearest source in undirected connected graphs. The finite time convergence of this protocol with communication delays and separately for asynchronous communication were studied  in  \cite{zhang2017perturbing}. Further, \cite{mo2019robustness}  proved that the estimation error of the protocol is ultimately bounded under additive noise. In this section, 
we  show that the biased min-consensus protocol is globally expISS under simultaneous manifestation of  communication delays, asynchronous communication and additive noise. To this end we  leverage the Razumikhin-type ISS Lyapunov based small gain theorem  of Section \ref{sec:main}. Unless explicitly mentioned, all proofs in this section are in the Appendix.

\subsection{Preliminaries}
The biased min-consensus protocol considers undirected,  connected graphs $\mathcal{G} = (N,E)$ with $N = \{1,2,\cdots, n \}$ the set of nodes and $E$ the set of edges. We call node $i$ a neighbor of node $j$ if there is an edge between $i$ and $j$. Further, $\mathcal{N}(i)$  denotes the set of all neighbors of node $i$. The presence of an edge indicates the existence of a communication link between nodes. We define $w_{ij} > 0$ as the edge weight/length between nodes $i$ and  $j$. Moreover, $i \in \mathcal{N}(j)$ implies $j \in \mathcal{N}(i)$ and $w_{ij} = w_{ji}$ as $\mathcal{G}$ is undirected. A path in $\mathcal{G}$ from $i_0$ to $i_h$ is the ordered set $P_{i_0i_h} = \{i_0, i_1, \cdots, i_h\}$ with $i_{k-1} \in \mathcal{N}(i_k)$ for all $k \in \{1,2,\cdots, h\}$, and the summation of weights of the constituent edges forms the length of the path.
We define $S \subsetneq N$ as the set of sources in $\mathcal{G}$. 


\subsection{Algorithms}
According to the Bellman's principle of optimality \cite{bellman1958routing}, $d_i$, the length of the shortest path between node $i$ and its nearest source obeys
\begin{equation}\label{eq:true}
d_i= \begin{cases}
	\min_{j\in \mathcal{N}(i)}\left \{ d_j+w_{ij}\right \}& i\notin S \\
	0 & i\in S
\end{cases} .
\end{equation}
The following definition characterizes the relation between node $i$ and node $j$ in (\ref{eq:true}).
\begin{definition}\label{def:true}
We define the minimizing $j$ in the first bullet of (\ref{eq:true}) as the true constraining node of $i$. As a node may have multiple true constraining nodes. The set of true constraining
nodes of a node $i \in N\setminus S$ is denoted as $\mathcal{C}(i)$. In particular, a source node does not have any true constraining node.
\end{definition}
As numbers of nodes and edges are both finite, there must exist a $\zeta \in (0,1)$ such that for all $i \in N$ with $\mathcal{N}(i)\setminus\mathcal{C}(i) \neq \emptyset$
\begin{equation}\label{eq:zeta1}
     \frac{ d_i }{ d_l + w_{il} } \leq \zeta, ~ \forall l \in \mathcal{N}(i)\setminus\mathcal{C}(i).
\end{equation}

The effective diameter is defined as follows.
\begin{definition}\label{def:eff}
Consider any sequence
of nodes such that the predecessor of each node is one of its
true constraining nodes. Define $\mathcal{D}$, the effective diameter
of $\mathcal{G}$, as the longest length such a sequence can have in $\mathcal{G}$. In particular, $\mathcal{D}$ has been proven to be finite \cite{mo2019robustness}.
\end{definition}
Let $k = 0$ be the initial time. Define $\dh_i(k)$ as the estimated length from node $i$ from the source set at time $k$. Without time delays and noises, $\dh_i(k + 1)$ in the biased min-consensus protocol obeys
\begin{equation}\label{eq:alg}
\dh_i(k+1)= \begin{cases}
	\min_{j\in \mathcal{N}(i)}\left \{ \dh_j(k)+w_{ij} \right \}& i\notin S \\
	0 & i\in S
\end{cases}.
\end{equation}
In such a protocol, the length estimate of each source is anchored at 0, while each non-source node computes the length estimate by iteratively using its neighbors' previous length estimates and the edge weights in between.

We consider three types of perturbations simultaneously on the protocol as described in (\ref{eq:alg}). First, we permit time delays in the exchange of distance estimates \cite{xia2006inference}. Second, we permit noise in the communication channel over which these exchanges occur, with the \emph{de facto} effect of variations in  edge weights. Such noise includes but is not limited to additive noise \cite{mo2019robustness} or quantization effects \cite{carli2007average}. Finally, we do not assume a central synchronizing clock in the environment, i.e., each node communicates asynchronously \cite{qin2012stationary}. Then (\ref{eq:alg}) can be  interpreted as 
\begin{flalign}\label{eq:algp}
& \dh_i(k+1) =  \nonumber\\
& \begin{cases}
	\underset{j\in \mathcal{N}(i)}{\min} \big\{ \dh_j(k - \tau_{ij}(k))  \\
 + w_{ij}(k- \tau_{ij}(k)) \big\} & i\in (N \setminus S)\cap U(k + 1)  \\
	0 & i \in S \\
	\dh_i(k) & \mathrm{otherwise}
\end{cases}
\end{flalign}
where $\tau_{ij}(k) \in \{0, 1, \cdots, \bar{\tau} \}$ denotes the bounded communication delay between node $i$ and node $j$ at time $k+1$, with $\bar{\tau}$ denoting the maximum communication delay and $\tau_{ij}(k) = 0$ indicating that there is no communication delay between node $i$ and node $j$, $w_{ij}(k)$ denotes the bounded, asymmetric and time-varying edge weight such that $0 < w_{ij}(k) \leq w_{\max}$ and $w_{ij}(k) \neq w_{ji}(k)$, and $U(k)$ denotes the set of nodes which update at time $k$, reflecting the asynchronous communication in the network. The well known principle of channel reciprocity, \cite{Wideband}, ensures that the communication delay between two nodes is identical in either direction. 

In principle asynchrony may prevent a node from updating its distance at all. In this paper we preclude that possibility by adopting a \emph{reasonable} model of asynchronous updates through the following assumption ensuring that at every $k$ all nodes update within a bounded window.
\begin{assumption}\label{ass:delta}
\cite{qin2012stationary} For every $k$, there exists a nonnegative integer $\delta$ such that the set of updating nodes $\bigcup_{k}^{k + \delta}U(k) = N$ for all $k \in \mathbb{Z}_+$.
\end{assumption}
Let $q_i(k) = \min\{j \in \mathbb{Z}_{+}: i \in U(k + 1 - j) \}$. From Assumption \ref{ass:delta}, $q_i(k) \leq \delta$ for all $i \in N\setminus S$. Recall that $\dh_i(k) = 0$ for all $k \in \mathbb{Z}_+$ and all $i \in S$, and $\tau_{ij}(k)$ in (\ref{eq:algp}) obeys $\tau_{ij}(k) \leq \bar{\tau}$ for all $i \in N$ and all $j \in \mathcal{N}(i)$. Given that $k \geq \delta + \bar{\tau}$, (\ref{eq:algp}) can be further written as
\begin{flalign}
&\!\!\! \dh_i(k+1) = \nonumber \\
&\!\!\!\begin{cases}
\underset{j\in \mathcal{N}(i)}{\min} \left \{ \dh_j(k\! -\! \tauh_{ij}(k))\!+\!w_{ij}(k\! -\! \tauh_{ij}(k)) \right \} & i\notin S \\
0 & S 
\end{cases}, \label{eq:algpr}
\end{flalign}
with $\tauh_{ij}(k)$ obeying 
\begin{equation}\label{eq:tauh}
    \tauh_{ij}(k) = \tau_{ij}(k - q_i(k)) + q_i(k) \leq \delta + \bar{\tau},
\end{equation}
As can be seen from (\ref{eq:algpr}), if there are no communication delay, asynchronous communication and the noise on the edge weight, i.e., $q_i(k) = 0$, $\tau_{ij}(k) = 0$ and $w_{ij}(k) = w_{ij}$, then (\ref{eq:algpr}) reduces to (\ref{eq:alg}).

With (\ref{eq:algpr}), we introduce the following definition.
\begin{definition}\label{eq:current}
We call the minimizing $j$ in (\ref{eq:algpr}) the constraining node of $i$ at time $k + 1$.
\end{definition}
\begin{figure}
	\centering
	{\includegraphics[width = 1\columnwidth]{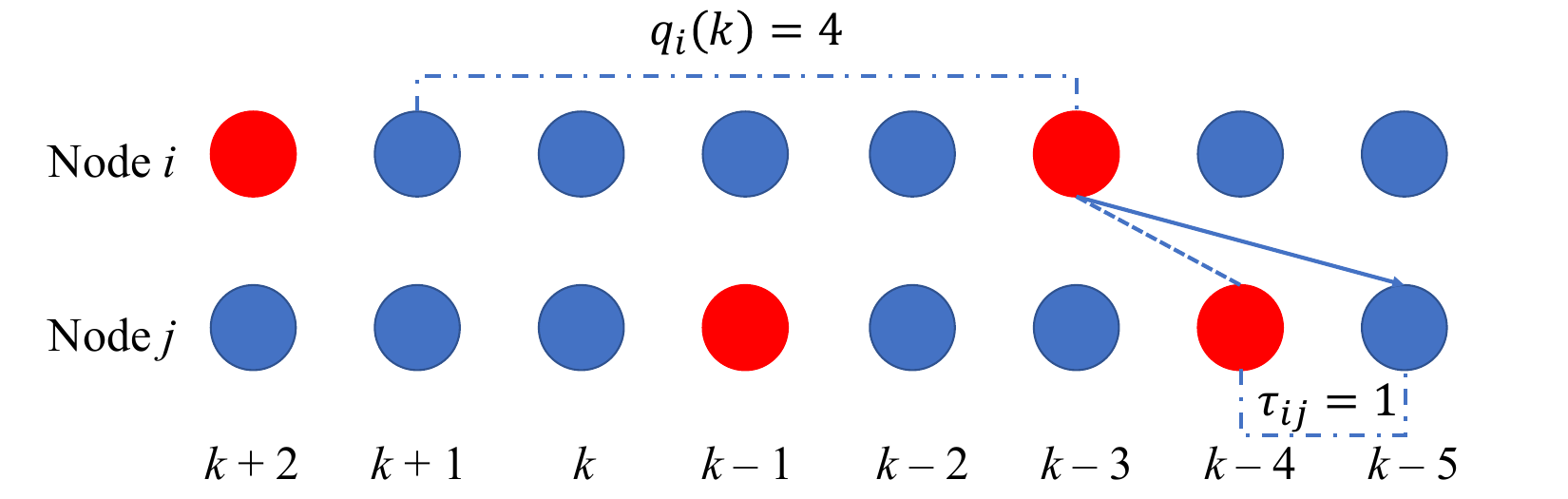}}
	\caption{Illustration of the update mechanism of the biased min-consensus protocol under time delays and communication asynchrony. In this example, the red circle indicates the node updates at this time step while the blue circle indicates the node does not update. The length estimate of node $i$ at time $k+1$ actually uses the length estimate of node $j$ in 6 time steps before, due to the one step time delay between $i$ and $j$ at time $k -3$, i.e., $\tau_{ij}(k - 4) = 1$, and the asynchronous communication $q_i(k) = 4$.
	}
	\label{fig:example}
\end{figure} 
Figure \ref{fig:example} illustrates the update mechanism of the biased min-consensus protocol under time delays and communication asynchrony. In the example, $\delta$ introduced in Assumption \ref{ass:delta} obeys $\delta = 4$ and the maximum communication delay $\bar{\tau} = 3$. As can be seen from Figure \ref{fig:example}, node $i$ does not update at time $k+1$ and $q_i(k) = 4$, leading to $\dh_i(k+1) = \dh_i(k - 3)$. Let node $j$ be the constraining node of node $i$ at time $k+1$, due to that the communication delay between $i$ and $j$ at time $k - 3$ obeys $\tau_{ij}(k - 4) = 1$, $\dh_i(k+1) = \dh_i(k - 3) = \dh_j(k - 5)+w_{ij}(k - 5)$.

The main assumption in this section is as follows.
\begin{assumption}\label{ass:main}
The underlying graph $\mathcal{G}$ is connected and undirected. The perturbed edge weight $w_{ij}(k - \tauh_{ij}(k))$ in (\ref{eq:algpr}) is positive and bounded, and unless mentioned otherwise, $k_0 = 0$ is the initial time. Furthermore,  $\dh_i(0) \geq 0$ for all $i \in N\setminus S$ and $\dh_i(0) = 0$ for $i \in S$.
\end{assumption}
\begin{remark}\label{re:initial}
	The requirement on  initial states in Assumption \ref{ass:main} is only  for simplifying the stability analysis, and  is not a strong assumption as it is shown in Lemma 1 in \cite{mo2018} that there exists a fintie $T$ such that $\dh_i(k) \geq 0$ for all $i \in N$ and $k \geq T$, regardless of the initial states. Further, a simple induction proof on $k$ can prove that $\dh_i(k)$ defined in (\ref{eq:algpr}) obeys $\dh_i(k) \geq 0$ for all $i \in N$ and $k \in \mathbb{Z}_+$.
\end{remark}

\subsection{Stability analysis}
We first transform the perturbed biased min-consensus protocol described by (\ref{eq:algp}) into the interconnected discrete time time-delay nonlinear system in the form of (\ref{eq:composite}). To this end, define the state of (\ref{eq:algp}) as
\begin{equation}\label{eq:state}
\xh(k) := [\xh_1(k), \xh_2(k), \cdots, \xh_n(k)] \in \mathbb{R}^n,
\end{equation}
where $\xh_i(k)$ obeys
\begin{equation}\label{eq:sstate}
    \xh_i(k) = \dh_i(k) - d_i,
\end{equation}
representing the estimation error of node $i$ at time $k$. We further take deviations of edge weights from their nominal values as the input. Define 
\begin{equation}\label{eq:sinput}
\uh_{ij}(k) = w_{ij}(k) - w_{ij}
\end{equation}
as the deviation of $w_{ij}$ at time $k$. Then the input of (\ref{eq:algp}) is given as
\begin{equation}\label{eq:input}
\uh(k) =  (\uh_{ij}(k))_{i \in N, j \in \mathcal{N}(i)} \in \mathbb{R}^{2|E|},
\end{equation}
Consider (\ref{eq:algp}). Clearly, for $i \in S$,  $\xh_i(k) = \xh_i(0) = 0$ for all $k \in \mathbb{Z}_+$. For $i \in N\setminus S$, we consider two cases: 1) $q_i(k - 1) > k - 1$, i.e., node $i$ has never updated yet at time $k$, then there holds $\xh_i(k) = \xh_i(0)$; 2) otherwise, with $\tauh_{ij}(k)$ defined in (\ref{eq:tauh}), it follows from (\ref{eq:algpr}) that
\begin{flalign}
&\xh_i(k) = \dh_i(k) - d_i  \nonumber \\
&= \underset{j\in \mathcal{N}(i)}{\min} \{ \dh_j(k - 1 - \tauh_{ij}(k - 1)) - d_j \nonumber \\ 
&~~~+ w_{ij}(k - 1 - \tauh_{ij}(k-1))
 - w_{ij} - 
 d_i + d_j + w_{ij} \} \nonumber \\
&= \underset{j\in \mathcal{N}(i)}{\min} \{ \xh_j(k - 1 - \tauh_{ij}(k-1)) \nonumber \\
&~~~+ \uh_{ij}(k - 1 - \tauh_{ij}(k - 1)) - d_i + d_j + w_{ij} \} \label{eq:cons}
\end{flalign}
As $d_i,d_j$ and $w_{ij}$ in (\ref{eq:cons}) are structural parameters of graph $\mathcal{G}$, for $i \in \{1, \cdots, n\}$ there holds
\begin{flalign}\label{eq:sys}
\!\xh_i(k \!+\! 1) \!= \! g_i(\xh_1(k \!- \!\tauh_{i1}(k)),\cdots, \xh_n(k\! -\! \tauh_{in}(k)), \uh(k\! - \!d) ) 
\end{flalign}
where $g_i: \mathbb{R}^n\times\mathbb{R}^{2|E|}  \rightarrow \mathbb{R}$ and $\uh(k - d) = (\uh_{ij}(k) - \tauh_{ij}(k))_{i \in N, j \in \mathcal{N}(i)} \in \mathbb{R}^{2|E|}$. From (\ref{eq:true}) and (\ref{eq:cons}), it can be verified that $g_i(0, \cdots, 0) = 0$. In particular, $g_i$ obeys
\begin{flalign}
&\!\!\!\!g_i(\xh_1(k \!- \! \tauh_{i1}(k)),\cdots, \xh_n(k \!-\! \tauh_{in}(k)), \uh(k \!-\! d) ) = \nonumber \\
&\!\!\!\!\begin{cases}
0, & i \in S \\
\underset{j\in \mathcal{N}(i)}{\min} \{ \xh_j(k \! -\! \tauh_{ij}(k))\! -\! d_i \!+ \! \\
\uh_{ij}(k \!- \!\tauh_{ij}(k)) \!+\! d_j \!+\! w_{ij} \}, & i \notin S ~\mathrm{and}~ q_i(k) \!\leq\! k \\
\xh_i(0),& \mathrm{otherwise}\label{eq:map}
\end{cases}.
\end{flalign}
Then the composite map can be defined as 
\begin{equation}\label{eq:compositemap}
    \xh(k+1) = \hat{G}(\xh_{[k - \delta - \bar{\tau}, k]}, \uh(k - d)), \forall k \geq \delta + \bar{\tau},
\end{equation}
where $\xh_{[k - \delta - \bar{\tau}, k]} \in (\mathbb{R}^n)^{\bar{\tau} + \delta + 1}$, $\hat{G}: (\mathbb{R}^n)^{\delta + \bar{\tau} + 1}\times \mathbb{R}^{2|E|} \rightarrow \mathbb{R}^n$ and $\hat{G}(0_{[k - \delta - \bar{\tau}, k]}, 0) = 0$. 

By interpreting (\ref{eq:algpr}) as an interconnected discrete time system of the form in (\ref{eq:sys}), we will show each node admits a Lyapunov-like function as defined in Assumption \ref{ass:vslf}.
\begin{lemma}\label{le:pra}
Suppose Assumptions \ref{ass:delta} and \ref{ass:main} hold. Consider (\ref{eq:sys}) with the input and state defined in (\ref{eq:input}) and (\ref{eq:state}), respectively. Let $V_i(\cdot) = |\cdot|, ~ \forall i \in N.$ Then 
\begin{equation}\label{eq:pro1}
	\alpha_{i1}(|\xi_i|) \leq V_i(\xi_i) \leq \alpha_{i2}(|\xi_i|),~ \forall \xi_i \in \bR
\end{equation}
with $\alpha_{i1} = \alpha_{i2} = \mathrm{id}$, and for $k \geq \bar{\tau} + \delta$
\begin{equation}\label{eq:pro2}
	V_i\big(\xh_i(k\!+\!1)\big)\! \leq \!\lambda_{ij}\big(V_j( \xh_j(k \!- \!\tauh_{ij}(k)) ) \big)\! +\! \lambda_{iu}(||\uh||_\infty),
\end{equation}
with $\lambda_{ij} = \lambda_{iu} = \mathrm{id}$.

In particular, when $i \in S$, $\lambda_{ij}$ in (\ref{eq:pro2}) obeys $\lambda_{ij} = 0$. When $i \notin S$, $j$ in (\ref{eq:pro2}) is a true constraining node of $i$ if $\xh_i(k+1) \geq 0$, and $j$ is the current constraining node of $i$ at time $k + 1$ otherwise. Further,  if $i \notin S$, $\xh_i(k+1) < 0$, and $j$ in (\ref{eq:pro2}) is the current constraining node but not true constraining node of node $i$, then $\lambda_{ij}$ in (\ref{eq:pro2}) obeys $\lambda_{ij} = \lambda_{iu} = \zeta\mathrm{id}$, with $\zeta \in (0,1)$ defined in (\ref{eq:zeta1}).
\end{lemma}

 The next lemma  shows that (\ref{eq:sys}) is globally $\mathcal{K}$-bounded.  
\begin{lemma}\label{le:sysk}
Suppose Assumptions \ref{ass:delta} and \ref{ass:main} hold. Then for all $i \in N$, $g_i$ defined in (\ref{eq:sys}) obeys
\begin{equation}\label{eq:gi}
|g_i(\xi_1,\xi_2,\cdots,\xi_n, \mu)|_\infty \leq \bar{\omega}_1(|\xi|_\infty) + \bar{\omega}_2(||\mu||_\infty) 
\end{equation}
for all $\xi = [\xi_1,\cdots,\xi_n]^{\mathrm{T}} \in \mathbb{R}^n$ and all $\mu \in \mathbb{R}^{2|E|}$,
with $\bar{\omega}_1 = \bar{\omega}_2 = \mathrm{id}$.
\end{lemma}
With Lemma \ref{le:sysk}, it can be readily verified that the composite map of (\ref{eq:sys}), $\hat{G}: (\mathbb{R}^n)^{\delta + \bar{\tau} + 1}\times \mathbb{R}^{2|E|} \rightarrow \mathbb{R}^n$, introduced in (\ref{eq:compositemap}), obeys that for all $\xi \in (\mathbb{R}^n)^{\delta + \bar{\tau} + 1}$ and all $\mu \in \mathbb{R}^{2|E|}$,
\begin{equation}\label{eq:kbgh}
| \hat{G}(\xi, \mu) |_\infty \leq |\xi|_\infty + |\mu|_\infty.
\end{equation}

While Lemma \ref{le:pra} implies that Assumption \ref{ass:vslf} holds for (\ref{eq:sys}). The following lemma further furnishes the state estimate of each node.
\begin{lemma}\label{le:small}
Suppose Assumptions \ref{ass:delta} and \ref{ass:main} hold. Let $M = \mathcal{D}(\delta + \bar{\tau}) + ( \mathcal{D}- 1)$, with $ \mathcal{D}, \bar{\tau}$ and $\delta$ defined in Definition \ref{def:eff}, (\ref{eq:algp}) and Assumption \ref{ass:delta}, respectively. Then for all $k \geq M$ and all $i \in N$ there holds
\begin{equation}\label{eq:small1}
V_i\big(\xh_i(k+1)\big) \leq \bar{\lambda}_{ij}\Big(V_j\big( \xh_j(k - \theta ) \big) \Big) + \bar{\lambda}_{iu}(||\uh||_\infty).
\end{equation}
where $V_i := |\cdot|$, $j \in N$, $\theta \in \mathbb{Z}_{[ \mathcal{D} - 1, M]}$, $\zeta\mathrm{id} \geq \bar{\lambda}_{ij} \in \mathcal{K}_{\infty}$ with $\zeta$ defined in (\ref{eq:zeta1}), and $\ \mathcal{D}\mathrm{id} \geq\bar{\lambda}_{iu} \in \mathcal{K}_{\infty}$.
\end{lemma}
Lemma \ref{le:small} implies that Assumption \ref{ass:small} also holds for (\ref{eq:sys}). Then expISS of (\ref{eq:sys}) can be proved using Corollary \ref{corr:exp}.
\begin{theorem}\label{thm:expISS}
Suppose Assumptions \ref{ass:delta} and \ref{ass:main} hold, with $ \mathcal{D}, \delta$ and $\bar{\tau}$ defined in Definition \ref{def:eff}, Assumption \ref{ass:delta} and (\ref{eq:algp}), respectively. Let $k_0 = M$ with $M =  \mathcal{D}(\delta + \bar{\tau}) + ( \mathcal{D} - 1)$ be the initial time. Then (\ref{eq:sys}) is expISS.
\end{theorem}
\begin{proof}
Obviously (\ref{eq:compositemap}) holds for all $k \geq k_0 =  \mathcal{D}(\delta + \bar{\tau}) + ( \mathcal{D} - 1)$ and (\ref{eq:compositemap}) admits a solution of length $M \geq \delta + \bar{\tau}$ (per Definition \ref{def:length}). Further, Lemma \ref{le:sysk} and Lemma \ref{le:small} prove that (\ref{eq:sys}) is globally $\mathcal{K}$-bounded and Assumptions \ref{ass:vslf}-\ref{ass:small} hold for (\ref{eq:sys}), respectively. As $\bar{\lambda}_{ij}$ defined in (\ref{eq:small1}) obeys $\bar{\lambda}_{ij} \leq \zeta\mathrm{id}$ with $\zeta \in (0,1)$ defined in (\ref{eq:zeta1}), the linear $\mathcal{K}_{\infty}$ function $\sigma_i$ introduced in (\ref{eq:smallr}) can be chosen as 
\begin{equation}\label{eq:sigma}
	\sigma_i = \mathrm{id}, ~ \forall i \in \{1,\cdots, n\}.
\end{equation}
Then by defining $V(\xi) = \max_{i \in N} \{\sigma_i^{-1}(V_i(\xi_i))\} = \max_{i \in N} \{V_i(\xi_i)\}$, i.e., $V = |\cdot|_\infty$, with $\xi \in \mathbb{R}^n$ and $V_i = |\cdot|$ defined in (\ref{eq:pro2}), it follows from Theorem \ref{thm:first} that $V$ is a Razumikhin-type ISS Lyapunov function of (\ref{eq:sys}), and thus (\ref{eq:sys}) is ISS with $k_0$ as the initial time. Further, as $V = |\cdot|_\infty$ obeys $\underline{\alpha}(|\xi|_\infty) \leq V(\xi) \leq \bar{\alpha}(|\xi|_\infty)$ for all $\xi \in \mathbb{R}^n$ with $\underline{\alpha} = \bar{\alpha} = \mathrm{id}$ defined in (\ref{eq:uplower}), and by (\ref{eq:kbgh}) $\hat{G}$ is globally $\cK$-bounded with $\omega_1$ in (\ref{eq:kbound}) obeying $\omega_1 = \mathrm{id}$, it follows from Corollary \ref{corr:exp} that (\ref{eq:sys}) is also expISS with $k_0$ as the initial time.
\end{proof}
\begin{figure}
	\centering
	{\includegraphics[width = 1\columnwidth]{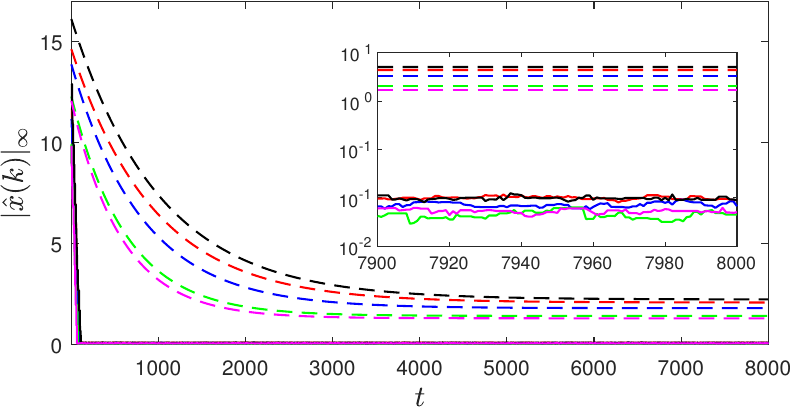}}
	\caption{Plot of the largest estimation error and its corresponding error bound for 5 runs of 500 nodes randomly distributed in a $4 \times 2~ \mathrm{km}^2$ area, communicating over a 0.35 km range, with asynchronous communication, communication delays and noisy measurement. The solid line represents $|\xh(k)|_\infty$, the largest estimation error among all nodes, while the dashed line in the same color represents the corresponding error bound. In particular, the partial enlarged view uses a base-10 logarithmic scale on the y-axis.}
	\label{fig:expISS}
\end{figure} 
\subsection{Refinement of the error bound}
In this subsection, we   give the upper bound of the estimation error of the biased min-consensus protocol in the form of (\ref{eq:expISS}) and (\ref{eq:rate}), with $k_0 = 0$ the initial time. 

From Theorem \ref{thm:expISS}, the linear $\mathcal{K}_{\infty}$ function $\sigma_i$ introduced in (\ref{eq:smallr}) in Theorem \ref{thm:first} is chosen as $\sigma_i = \mathrm{id}$ for all $i \in \{1,\cdots, n\}$. As $\bar{\lambda}_{ij}$ and $\bar{\lambda}_{iu}$ in (\ref{eq:small1}) obey $\bar{\lambda}_{ij} \leq \zeta\mathrm{id}$ and $\bar{\lambda}_{iu} \leq  \mathcal{D}\mathrm{id}$, respectively, it follows from Theorem \ref{thm:first} that the Razumikhin-type ISS Lyapunov function for (\ref{eq:sys}) obeys that for all $k \geq M$ with $M =  \mathcal{D}(\delta + \bar{\tau}) + ( \mathcal{D} - 1)$,
\begin{equation}\label{eq:Lyanew}
	V( \xh(k+1) ) \leq \max_{\theta \in \mathbb{Z}_{[k - M, k]}} \zeta V(\xh(\theta)) +  \mathcal{D}|| \uh ||_\infty
\end{equation}
with $V(\cdot) = |\cdot|_\infty$. Implementing (\ref{eq:small1}) in Lemma \ref{le:small} on (\ref{eq:Lyanew}) repeatedly until it cannot be applied, it
follows from (\ref{eq:uag})-(\ref{eq:ule1}) in Theorem \ref{the:suffISS} that $V(\xh(k+1))$ obeys
\begin{flalign}\label{eq:fu}
	V(\xh(k+1)) \leq \zeta^{\frac{k - j + 1}{M + 1}}( V(\xh(j)) ) + \frac{\mathcal{D}}{1 - \zeta} ||\uh||_\infty
\end{flalign}
where $j \in \mathbb{Z}_{[0, M]}$. Applying (\ref{eq:pro2}) in Lemma \ref{le:pra} repeatedly on $V(\xh(j))$ until it is related to the initial state, there holds
\begin{equation}\label{eq:fur}
    V(\xh(j)) \leq V\Big(\xh\big(j - \sum_{i = 1}^{q}(m_i + 1) \big)\Big) + q||\uh||_\infty
\end{equation}
where $m_i \in \mathbb{Z}_{[0, \delta + \bar{\tau}]}$, $l := j - \sum_{i = 1}^{q}(m_i + 1) \geq 0$ with $V(\xh(l)) = |\xh(l)|_\infty = V(\xh_i(l)) = V(\xh_i(0))$ (assume $i = \arg \max_{j \in N}\{\xh_j(l) \}$). Further, we have $q \leq \mathcal{D} - 1$ as otherwise it follows from Lemma \ref{le:small} that (\ref{eq:small1}) can continue to be applied. Putting (\ref{eq:fur}) into (\ref{eq:fu}), we can obtain
\begin{flalign}
	&V(\xh(k+1))  \nonumber\\
	&\!\leq \! \zeta^{\frac{k - j + 1}{M + 1}}\big( V(\xh_i(0)) + (\mathcal{D} - 1)||\uh||_\infty \big) +\frac{1}{1 - \zeta}\mathcal{D}(||\uh||_\infty) \nonumber\\
	&\!\leq \! \zeta^{\frac{k \!-\! M + 1}{M\! + \!1}}\big(|\xi|_\infty \!+\! (\mathcal{D} \!-\! 1)||\uh||_\infty \big) \!+\!\frac{1}{1\! -\! \zeta}\mathcal{D}||\uh||_\infty\label{eq:normt} \\
	&\!\leq \! \zeta^{\frac{k \!+\! 1}{M \!+\! 1}}\zeta^{\frac{-M}{M \!+\! 1}}(|\xi|_\infty)\!+\!\big(\zeta^{\frac{1}{M \!+\! 1}}(\mathcal{D} \!-\! 1)\! +\! \frac{1}{1 \!- \!\zeta}\mathcal{D}\big)||\uh||_\infty \label{eq:nt1}
\end{flalign}
where (\ref{eq:normt}) uses the fact that $j \in \mathbb{Z}_{[0, M]}$, $V = |\cdot|_\infty$ and $\xi \in \mathbb{R}^n$ is the initial state, and (\ref{eq:nt1}) results from $k \geq M$. With (\ref{eq:nt1}), for $k \geq M$, $|\xh(k+1)|_\infty$ can be characterized by
\begin{equation}\label{eq:exp1}
	|\xh(k + 1)|_\infty \leq \beta(| 
	\xi|_\infty, k + 1) + \bar{\lambda}_u(||\uh||_\infty)
\end{equation}
where $\xi \in (\mathbb{R}^n)^{M + 1}$ is the initial state, $\beta(r,t) = \bar{\zeta}^t\zeta^{\frac{-M}{M + 1}}r$ with $\bar{\zeta} = \zeta^{\frac{1}{M + 1}}$ and $\bar{\lambda}_u = (\zeta^{\frac{1}{M + 1}}( \mathcal{D} - 1) + \frac{1}{1 - \zeta} \mathcal{D})\mathrm{id} \in \mathcal{K}_{\infty}$.
\section{Simulation}\label{sec:simulations}
In this section, we verify the theoretical results in previous sections through simulations. We run our simulations in a $4\mathrm{km} \times 2\mathrm{km}$ area, where 500 nodes, including a source node, are randomly placed, communicating over a 0.35 km radius. The distance between nodes is measured in hop counts, i.e., in (\ref{eq:alg}) $w_{ij} = 1$ for all $i,j \in N$. We consider the asynchronous communication, communication delays and noisy measurement as the perturbations. The simulation result is shown in Figure \ref{fig:expISS}. In this case, the time-varying edge weight $w_{ij}(k)$ is randomly distributed between $[0.99, 1.01]$ due to the noisy measurement, $\delta$ in Assumption \ref{ass:delta} has an expectation of 2, and the maximum communication delay $\bar{\tau}$ is 2. We run the simulations for 5 trials, the effective diameter $ \mathcal{D}$ defined in Definition \ref{def:eff} and $\zeta$ defined in (\ref{eq:zeta1}) for those 5 trials are $14, 13, 12, 15, 11$ and $0.9286, 0.9231,	0.9091,	0.9286, 0.9091$, respectively, and the initial states for those 5 trials are randomly distributed between $[0, d_{\max}/2]$ where $d_{\max} = \max_{i \in N}\{d_i\}$ with $d_i$ the length of the shortest path from $i$ to the source node as defined in (\ref{eq:true}). As can be seen from Figure \ref{fig:expISS}, under the above perturbations, for all the trials, $|\xh(k)|_\infty$, the largest estimation error among all nodes (in solid line) will not converge to zero but drops exponentially fast below an upper bound (in dashed line with the same color) characterized by (\ref{eq:exp1}). However, due to the conservative nature of the Lyapunov-based approach, there is still a gap between the estimation error and its upper bound. While the estimation error rapidly drops to around 0.1 within 200 rounds, the error bound exponentially decreases to between 1 to 2. 

\section{Conclusion}\label{sec:conclusion}
In this paper, we present both Lyapunov and small gain approaches for the ISS of discrete time time-delay systems. Specifically the converse theorems of both approaches with respect to expISS of discrete time time-delay systems are also provided. By leveraging the proposed Lyapunov-based small gain theorem, we prove that the biased min-consensus protocol, which is used to compute the length of the shortest path from each non-source node to its nearest source, is globally exponentially input-to-state stable under three perturbations: 1) time-varying edge weights; 2) communication delay; and 3) asynchronous communication. Simulations are provided to verify the validity of the theoretical results. Our future work include two directions: 1) provide converse Lyapunov theorems for the ISS of discrete time time-delay systems; 2) establish the Krasovskii-type Lyapunov-based small gain theorems for the discrete time time-delay systems.

As with most results involving ISS, ours involves small gain like theorems. It is well known that there is an equivalence between the classical small gain and passivity theorems, \cite{anderson1972small}. Passivity in turn has been useful in proving stability of adaptive, \cite{anderson1986stability,dasgupta1986output}, and multiagent systems \cite{Passivity}. A future direction could also be to formulate passivity type theorems, perhaps by using a variation of dissipative type Lyapunov functions, e.g. by having  the inner product of input and output in the stead of their norms.

\bibliographystyle{automaticanew}        
\bibliography{refs}

\section*{Appendix}\label{App}

\noindent
{\bf Proof of Lemma \ref{le:pra}:}
Obviously $V_i$ satisfies (\ref{eq:pro1}) with $\underline{\alpha}_i = \bar{\alpha}_i = \mathrm{id}$. For $i \in S$, it follows from (\ref{eq:true}), (\ref{eq:algp}) and (\ref{eq:sstate}) that for all $k \in \mathbb{Z}_+$
$V_i(\xh_i(k + 1)) = | \dh_i(k + 1) - d_i | = 0,$
and thus $V_i$ satisfies (\ref{eq:pro2}) for $i \in S$ with $\lambda_{ij}$ in (\ref{eq:pro2}) obeying $\lambda_{ij} = 0$. For $i \notin S$, 
we consider two cases: 1) $\dh_i(k + 1) \geq d_i$; 2) $\dh_i(k + 1) < d_i$. In the former case, assume $j$ is the true constraining node of $i$ (per Definition \ref{def:true}), then it follows from (\ref{eq:algpr}) that for all $k \geq \bar{\tau} + \delta$, $V_i(\xh_i(k + 1))$ obeys 
\begin{flalign}
&V_i(\xh_i(k + 1)) = |\dh_i(k + 1) - d_i | = \dh_i(k + 1) - d_i \nonumber \\
& \leq   \dh_j(k - \tauh_{ij}(k) ) + w_{ij}(k - \tauh_{ij}(k)) - d_j - w_{ij} \label{eq:jit} \\
& \leq  |\dh_j(k - \tauh_{ij}(k)) - d_j| + |w_{ij}(k - \tauh_{ij}(k)) - w_{ij}| \nonumber \\
& \leq   \lambda_{ij}\big(V_j(\xh_j(k - \tauh_{ij}(k)))\big) + \lambda_{iu}(||\uh||_\infty) \label{eq:tif}
\end{flalign}
where (\ref{eq:jit}) comes from (\ref{eq:algpr}) and  (\ref{eq:true}), and the equality in (\ref{eq:jit}) holds if $j$ is also the constraining node of $i$ at time $k + 1$, and in (\ref{eq:tif}) $\lambda_{ij} = \lambda_{iu} = \mathrm{id}$. 

In the latter case, assume $j$ is the current constraining node of $i$ at time $k + 1$, we obtain
\begin{flalign}
&V_i(\xh_i(k + 1))  = |\dh_i(k + 1) - d_i |  =  d_i - \dh_i(k + 1) \nonumber \\ 
& \leq d_j + w_{ij} - \dh_j(k - \tauh_{ij}(k) ) - w_{ij}(k - \tauh_{ij}(k)) \label{eq:utr} \\
& \leq |\dh_j(k - \tauh_{ij}(k)) - d_j| + |w_{ij}(k - \tauh_{ij}(k)) - w_{ij}| \nonumber \\
& \leq   \lambda_{ij}\big(V_j(\xh_j(k - \tauh_{ij}(k)))\big) + \lambda_{iu}(||\uh||_\infty)\label{eq:fin} 
\end{flalign}
where (\ref{eq:utr}) uses (\ref{eq:true}), the equality in (\ref{eq:utr}) holds if $j$ is also a true constraining node of $i$, and in (\ref{eq:fin}) $\lambda_{ij} = \lambda_{iu} = \mathrm{id}$.
Specifically, if $j$ in (\ref{eq:utr}) is not a true constraining node of node $i$, from (\ref{eq:zeta1}), (\ref{eq:utr}) becomes
\begin{flalign}
	&V_i(\xh_i(k + 1)) \nonumber \\
	&\leq \zeta(d_j + w_{ij}) - \dh_j(k - \tauh_{ij}(k) ) - w_{ij}(k - \tauh_{ij}(k)) \label{eq:usezeta} \\
	& < \zeta(d_j + w_{ij}) - \zeta(\dh_j(k - \tauh_{ij}(k) ) + w_{ij}(k - \tauh_{ij}(k))) \label{eq:uis} \\
	&\leq \zeta|\dh_j(k - \tauh_{ij}(k)) - d_j| + \zeta|w_{ij}(k - \tauh_{ij}(k)) - w_{ij}| \nonumber \\
	& \leq   \lambda_{ij}\big(V_j(\xh_j(k - \tauh_{ij}(k)))\big) + \lambda_{iu}(||\uh||_\infty)\label{eq:spe}
\end{flalign}
where (\ref{eq:usezeta}) uses (\ref{eq:zeta1}), (\ref{eq:uis}) uses that fact that both $\dh_j(k - \tauh_{ij}(k))$ and $w_{ij}(k - \tauh_{ij}(k))$ are nonnegative (per Assumption \ref{ass:main} and Remark \ref{re:initial}), and in (\ref{eq:spe}) $\lambda_{ij} = \lambda_{iu} = \zeta\mathrm{id}$.

\noindent
{\bf Proof of Lemma \ref{le:sysk}:} 
According to (\ref{eq:algpr}), we consider three cases: 1) $i \in S$; 2) $i \notin S$ and $q_i(k) > k$ ($i$ has not updated yet); 3) $i \notin S$ and $q_i(k) \leq k$. For the first two cases, it follows from (\ref{eq:map}) that either $\xh_i(k + 1) = 0$ for all $k$ or $\xh_i(k + 1) = \xh_i(0)$ for some $k < \delta$, and thus (\ref{eq:gi}) holds trivially. For the last case, it follows from (\ref{eq:tif}) and (\ref{eq:fin}) in Lemma \ref{le:pra} that $\xh_i(k)$ obeys
$
|\xh_i(k + 1)|_\infty 
\leq |\xh_j(k - \tauh_{ij}(k))|_\infty + ||\uh||_\infty,
$
and thus our claim follows.

\noindent
{\bf Proof of Lemma \ref{le:small}:} From (\ref{eq:pro2}) in Lemma \ref{le:pra}, for $k \geq M$, there holds
\begin{flalign}
	&V_{i_0}\big(\xh_{i_0}(k+1)\big) \nonumber \\
	& \leq \lambda_{i_0i_1}\big(V_{i_1}( \xh_{i_1}(k - \tauh_{i_0i_1}(k)) ) \big) + \lambda_{i_0u}(||\uh||_\infty) \nonumber \\ 	
	& \leq \lambda_{i_0i_1}\Big( \lambda_{i_1i_2}\big( \lambda_{i_1u}(||\uh||_\infty) + V_{i_2}( \xh_{i_2}(k - \tauh_{i_0i_1}(k) \nonumber \\ 
	& ~~~  -\tauh_{i_1i_2}(k - 1 -  \tauh_{i_0i_1}(k))) \big)\Big) + \lambda_{i_0u}(||\uh||_\infty) \nonumber \\
	&  \cdots \nonumber \\
	& \leq \lambda_{i}( \mathcal{D}\! -\! 1)\big(V_{i_{ \mathcal{D}- 1}}(\xh_{i_{ \mathcal{D} \!-\! 1 }}(k \!-\! \theta_1)) \big) \!+\! \lambda_{iu}( \mathcal{D}\! - \!1)(||\uh||_\infty) \label{eq:bfin} \\
	& \leq \lambda_{i}( \mathcal{D})\big(V_{i_ \mathcal{D}}(\xh_{i_ \mathcal{D}}(k - \theta_2)) \big) + \lambda_{iu}( \mathcal{D})(||\uh||_\infty) \label{eq:fina}
\end{flalign}
where in (\ref{eq:bfin}) and (\ref{eq:fina}) functions $\lambda_{i}(\cdot)$ and $\lambda_{iu}(\cdot)$ obey 
\begin{equation}\label{eq:ll}
	\lambda_{i}(x) = \underset{l=0}{\overset{x - 1}{\mathrm{C}}} \lambda_{i_li_{l+1}} \in \mathcal{K}_{\infty}
\end{equation} 
and
\begin{equation}\label{eq:liu}
	\lambda_{iu}(x) = \sum_{m = 0}^{x - 2}\underset{l=0}{\overset{m}{\mathrm{C}}}\lambda_{i_li_{l+1}}\circ\lambda_{i_{l+1}u} + \lambda_{i_0u} \in \mathcal{K}_{\infty},
\end{equation}
respectively, and $\theta_1 \in [ \mathcal{D} - 2, M - (\delta + \bar{\tau} + 1)], \theta_2 \in [ \mathcal{D} - 1, M]$. As $\lambda_{i_li_{l+1}}$ in (\ref{eq:ll}) obeys $\lambda_{i_li_{l+1}} \leq \mathrm{id}$ for all $l \in \{0,1,\cdots,  \mathcal{D} - 1\}$ by Lemma \ref{le:pra}, then it follows from (\ref{eq:liu}) and (\ref{eq:pro2}) that $\lambda_{iu}( \mathcal{D})$ in (\ref{eq:fina}) obeys $\lambda_{iu}( \mathcal{D}) \leq  \mathcal{D}\mathrm{id}$.

We next prove that $\lambda_i( \mathcal{D})$ in (\ref{eq:fina}) obeys $\lambda_i( \mathcal{D}) \leq \zeta\mathrm{id}$. We prove our claim by contradiction. Suppose $\lambda_i( \mathcal{D}) > \zeta\mathrm{id}$, as $\lambda_{i_li_{l+1}}$ obeys $\lambda_{i_li_{l+1}} = \mathrm{id}$ or $\lambda_{i_li_{l+1}} \leq \zeta\mathrm{id}$ for all $l \in \{0,1,\cdots,  \mathcal{D} - 1\}$ by Lemma \ref{le:pra}, then for all $l \in \{0,1,\cdots,  \mathcal{D} - 1\}$, there holds $\lambda_{i_li_{l+1}} = \mathrm{id}$, which further implies $i_{l+1}$ is a true constraining node of $i_l$ by Lemma \ref{le:pra}. However, by Definition \ref{def:eff}, the length of such a sequence can not exceed $ \mathcal{D}$, then $i_{ \mathcal{D}}$ can not be a true constraining node of $i_{ \mathcal{D} - 1}$. We consider two cases: 1) $\xh_{i_ \mathcal{D}}(k - \theta_2)$ in (\ref{eq:bfin}) obeys $\xh_{i_ \mathcal{D}}(k - \theta_2) \geq 0$; 2) $\xh_{i_ \mathcal{D}}(k - \theta_2) < 0$. In the former case, as $i_{ \mathcal{D}}$ can not be the true constraining node of $i_{ \mathcal{D} - 1}$, it follows from Lemma \ref{le:pra} that $i_{ \mathcal{D} - 1}$ is the source node and $\lambda_{i_{ \mathcal{D} - 1}i_{ \mathcal{D}}} = 0$. In the latter case, it follows from Lemma \ref{le:pra} that $\lambda_{i_{ \mathcal{D} - 1}i_{ \mathcal{D}}} = \zeta\mathrm{id}$, establishing the contradiction. Thus our claim follows, and $\lambda( \mathcal{D})$ in (\ref{eq:fina}) obeys $\lambda( \mathcal{D}) \leq \zeta\mathrm{id}$, completing our proof.
\end{document}